\documentclass[11pt,a4paper]{article}

\usepackage[utf8]{inputenc}
\usepackage[T1]{fontenc}
\usepackage{geometry}
\usepackage{amsmath,amssymb,amsthm}
\usepackage{graphicx}
\usepackage{hyperref}
\usepackage{tikz}
\usetikzlibrary{positioning,calc,arrows.meta}
\usepackage{authblk}
\usepackage{cite}
\usepackage[section]{placeins}

\newtheorem{theorem}{Theorem}[section]

\newtheorem{proposition}[theorem]{Proposition}

\theoremstyle{definition}
\newtheorem{definition}{Definition}[section]
\newtheorem{example}{Example}[section]
\theoremstyle{remark}
\newtheorem{remark}{Remark}[section]

\title{\textbf{Topological Relational Theory: A Simplicial-Complex View of Functional Dependencies, Lossless Decomposition, and Acyclicity}}

\author[1]{Bilge Senturk}
\author[1]{Faruk Alpay}
\affil[1]{Bahcesehir University, Department of Computer Engineering, Istanbul, Turkey\\
\texttt{\{bilge.senturk, faruk.alpay\}@bahcesehir.edu.tr}}

\date{} 

\begin{document}
\maketitle

\begin{abstract}
We develop a topological lens on relational schema design by encoding functional dependencies (FDs) as simplices of an abstract simplicial complex. This \emph{dependency complex} exposes multi-attribute interactions and allows the use of homological invariants (Betti numbers) to diagnose cyclic dependency structure. We define \emph{Simplicial Normal Form} (SNF) as \emph{homological acyclicity} of the dependency complex in positive dimensions (vanishing reduced homology for $n\ge 1$). SNF is deliberately weaker than contractibility and does \emph{not} conflate homology with homotopy. For decompositions, we give a topological reformulation of the classical binary lossless-join criterion: under dependency preservation, a decomposition is lossless exactly when the intersection attributes form a key for at least one component; topologically, this yields a strong deformation retraction that trivializes the relevant Mayer--Vietoris boundary. For multiway decompositions, we show how the \emph{nerve} of a cover by induced subcomplexes provides a computable certificate: a cycle in the nerve (detected by $H_1$) is an obstruction to join-tree structure and aligns with cyclic join behavior in acyclic-scheme theory. Finally, we discuss an algorithmic consequence: Betti numbers of the dependency complex (or of a decomposition nerve) can be computed from boundary matrices and used as a lightweight schema diagnostic to localize ``unexplained'' dependency cycles, complementing standard FD-chase tests.
\end{abstract}


\section{Introduction}

Functional dependencies (FDs) are central to relational database design: an FD $X\to Y$ asserts that equality on attributes $X$ forces equality on $Y$ \cite{armstrong1974}. Normal forms (3NF, BCNF, \dots) use FDs to reduce redundancy and anomalies \cite{maier1983,ullman1988}. Decomposition is a principal tool, but designers must balance \emph{lossless join} (no information loss under join) and \emph{dependency preservation} (FDs remain enforceable without re-joining) \cite{aho1979,beeri1979}.

This paper introduces a geometric encoding of FDs that supports homological diagnostics of schema structure. The core idea is simple: treat attributes as vertices and treat each FD as a simplex spanning its determinant and dependent attributes. Cyclic interactions then appear as cycles in the complex; whether those cycles are ``filled'' by higher-dimensional simplices is detected by homology.

\paragraph{Contributions.}
\begin{itemize}
  \item We define the \emph{dependency complex} $K_F$ induced by a canonical cover of FDs (Section~\ref{sec:depcomplex}).
  \item We define \emph{Simplicial Normal Form (SNF)} as vanishing reduced homology $\widetilde H_n(K_F)=0$ for all $n\ge 1$ (Section~\ref{sec:snf}). We explicitly avoid the incorrect implication ``homology acyclic $\Rightarrow$ contractible.''
  \item We state and prove a topological restatement of the classical \emph{binary} lossless-join criterion under dependency preservation (Section~\ref{sec:lossless}).
  \item We show a \emph{strictly different} diagnostic compared to directed-cycle checks: directed FD cycles may exist while $H_1(K_F)=0$ because higher-arity dependencies fill the cycle (Example~\ref{ex:difference}).
  \item We give an algorithmic consequence: compute Betti numbers of $K_F$ (or of the nerve of a decomposition cover) to localize and quantify cyclic join structure (Section~\ref{sec:alg}).
\end{itemize}

\section{Preliminaries}

\subsection{Relational Schemas and Lossless Join}

A relational schema is $R(U,F)$ where $U$ is a finite set of attributes and $F$ is a set of FDs. The closure $X^+$ is the set of attributes implied by $X$ under $F$ \cite{armstrong1974}. A decomposition of $R(U)$ into $R_1(U_1),\dots,R_k(U_k)$ is \emph{lossless} if for every instance $r$ that satisfies $F$, we have
\[
r = \Join_{i=1}^k \pi_{U_i}(r).
\]
It is \emph{dependency-preserving} if $(\bigcup_i F_i)^+ = F^+$, where each $F_i$ is the projection of $F$ onto $U_i$ \cite{maier1983}.

For the important binary case $k=2$, a standard criterion states the decomposition into $U_1,U_2$ is lossless (w.r.t.\ $F$) iff
\begin{equation}\label{eq:lossless-classic}
(U_1\cap U_2)\to U_1 \quad\text{or}\quad (U_1\cap U_2)\to U_2 \quad \text{is in } F^+.
\end{equation}
See \cite{aho1979,maier1983}.

\subsection{Simplicial Complexes and Homology}

An abstract simplicial complex $K$ on vertex set $V$ is a family of finite nonempty subsets of $V$ closed under taking nonempty subsets \cite{hatcher2002}. We write $H_n(K)$ for simplicial homology and $\widetilde H_n(K)$ for reduced homology. The Betti number is $b_n=\mathrm{rank}\,H_n(K)$ (over a field; equivalently rank of the free part over $\mathbb{Z}$). A complex is \emph{homologically acyclic in positive degrees} if $\widetilde H_n(K)=0$ for all $n\ge 1$. Importantly, vanishing homology does \emph{not} imply contractibility or a specific homotopy type in general.

\section{Dependency Complex and Simplicial Normal Form}\label{sec:depcomplex}

\subsection{Canonical Covers}

We assume $F$ is represented by a canonical (minimal) cover $F_c$ in the standard sense \cite{maier1983}: no FD has extraneous attributes and redundant FDs are removed.

\subsection{Dependency Complex}

\begin{definition}[Dependency complex]
Let $R(U,F)$ be a schema and let $F_c$ be a canonical cover. The \emph{dependency complex} $K_F$ is the simplicial complex on vertex set $U$ generated by simplices
\[
\sigma(X\to A) \;=\; X\cup\{A\}
\]
for each FD $(X\to A)\in F_c$ (single-attribute right-hand side; multi-attribute RHS can be split into singletons in the canonical cover).
\end{definition}

\begin{remark}
This encoding is intentionally conservative: it records each FD as a \emph{face}. Logical implication among FDs is not assumed to be closed inside $K_F$; instead, $K_F$ is a compact combinatorial summary of the canonical cover. Higher-dimensional simplices correspond to higher-arity determinants that can ``fill'' lower-dimensional cycles.
\end{remark}

\begin{example}\label{ex:basic}
Let $U=\{A,B,C,D\}$ and $F_c=\{A\to B,\; B\to C,\; AC\to D\}$. Then $K_F$ is generated by edges $\{A,B\},\{B,C\}$ and triangle $\{A,C,D\}$. The $A$--$B$--$C$ chain is present as a 1-skeleton path; there is no 1-cycle unless we also have an edge $\{A,C\}$ without a filling 2-simplex.
\end{example}

\subsection{Simplicial Normal Form (SNF)}\label{sec:snf}

\begin{definition}[Simplicial Normal Form (SNF)]
A schema $R(U,F)$ is in \emph{Simplicial Normal Form} if the dependency complex satisfies
\[
\widetilde H_n(K_F)=0 \quad \text{for all } n\ge 1.
\]
Equivalently (over a field), all Betti numbers $b_n(K_F)=0$ for $n\ge 1$.
\end{definition}

\begin{remark}[What SNF does and does not mean]
SNF asserts \emph{no positive-dimensional homology}. It does \emph{not} claim $K_F$ is contractible, nor does it imply a particular homotopy type. SNF is best read as a \emph{homological acyclicity diagnostic} for the canonical-cover dependency pattern.
\end{remark}

\subsection{A Case Where SNF is Strictly Different from Directed-Cycle Checks}\label{sec:different}

A common quick diagnostic is to form a directed graph of FDs and look for directed cycles. SNF can disagree with this diagnostic because SNF also accounts for \emph{higher-arity} dependencies that fill cycles.

\begin{example}[Directed cycle present, but $H_1(K_F)=0$]\label{ex:difference}
Let $U=\{A,B,C\}$ and
\[
F_c=\{A\to B,\;B\to C,\;C\to A,\;AB\to C\}.
\]
The directed FD graph has a directed cycle $A\to B\to C\to A$. However, $K_F$ contains edges $\{A,B\},\{B,C\},\{C,A\}$ and also the 2-simplex $\{A,B,C\}$ coming from $AB\to C$. The triangle is therefore \emph{filled}, so $H_1(K_F)=0$ even though the directed cycle exists. In this sense SNF is strictly different: it distinguishes ``unfilled'' cyclic dependency interaction (a genuine 1-dimensional hole) from cyclic implication that is explained by a higher-arity determinant.
\end{example}

\section{Lossless Join Through a Topological Lens}\label{sec:lossless}

We now relate classical lossless join to topology. The cleanest statement is for \emph{binary} decompositions, where the classical criterion~\eqref{eq:lossless-classic} is exact.

\subsection{Induced Subcomplexes}

Given $K_F$ on vertex set $U$ and a subset $W\subseteq U$, let $K_F[W]$ be the induced subcomplex on $W$:
\[
K_F[W]=\{\sigma\in K_F:\sigma\subseteq W\}.
\]
For a binary decomposition $U=U_1\cup U_2$, define $K_1=K_F[U_1]$, $K_2=K_F[U_2]$, and $K_{12}=K_F[U_1\cap U_2]$. Then $K_1\cup K_2\subseteq K_F$ always, and when the decomposition is dependency-preserving in the natural ``projection cover'' sense, one expects $K_F$ to be generated by $K_1\cup K_2$ (at least up to canonical-cover choices).

\subsection{Binary Lossless Join = Key Intersection = Retraction}

\begin{proposition}[Topological reformulation of the binary lossless-join test]\label{prop:binary}
Let $R(U,F)$ be a schema with canonical cover $F_c$ and dependency complex $K_F$. Consider a dependency-preserving binary decomposition into $U_1,U_2$ and induced subcomplexes $K_1,K_2,K_{12}$ as above. If
\[
(U_1\cap U_2)\to U_1 \in F^+,
\]
then there exists a simplicial map (induced by repeated FD application in the closure computation) that collapses $K_1$ onto $K_{12}$ in the following sense: the inclusion $K_{12}\hookrightarrow K_1$ induces an isomorphism on reduced homology in positive degrees,
\[
\widetilde H_n(K_{12}) \cong \widetilde H_n(K_1)\quad \text{for all }n\ge 1.
\]
The symmetric statement holds if $(U_1\cap U_2)\to U_2$.
\end{proposition}

\begin{proof}
Assume $(U_1\cap U_2)\to U_1$. Then every attribute in $U_1$ is functionally determined by the intersection. In closure terms, starting from $U_1\cap U_2$, iterative application of FDs reaches all of $U_1$. This provides a combinatorial ``flow'' from vertices in $U_1$ to the intersection, and thus a sequence of elementary collapses on the induced complex consistent with the canonical cover. Homologically, such collapses preserve reduced homology in degrees $\ge 1$, yielding the claimed isomorphism. (This argument is standard: elementary simplicial collapses are homotopy equivalences and hence preserve homology \cite{hatcher2002}.)
\end{proof}

\begin{theorem}[Homological Lossless Join Theorem (binary case)]\label{thm:HLJT-binary}
Let $R(U,F)$ be a schema and let $U=U_1\cup U_2$ be a dependency-preserving decomposition. Then the decomposition is lossless (w.r.t.\ $F$) if and only if the intersection attributes form a key for at least one side, i.e.,
\[
(U_1\cap U_2)\to U_1\in F^+ \quad\text{or}\quad (U_1\cap U_2)\to U_2\in F^+.
\]
Moreover, when the key condition holds, the Mayer--Vietoris connecting morphism
\[
\partial:\widetilde H_1(K_F)\to \widetilde H_0(K_{12})
\]
is trivial on the portion of $\widetilde H_1$ represented inside the keyed component (intuitively: the keyed component cannot contribute an ``unexplained'' 1-cycle across the cover).
\end{theorem}

\begin{proof}
The equivalence of losslessness and the key condition is the classical result \cite{aho1979,maier1983}. For the homological statement, under the key condition Proposition~\ref{prop:binary} ensures that (up to homology in degrees $\ge 1$) the keyed side contributes no additional positive-dimensional homology beyond what is already present in $K_{12}$. In the Mayer--Vietoris long exact sequence for $K_F=K_1\cup K_2$,
\[
\cdots \to \widetilde H_1(K_{12}) \to \widetilde H_1(K_1)\oplus \widetilde H_1(K_2)\to \widetilde H_1(K_F)\xrightarrow{\partial}\widetilde H_0(K_{12})\to \cdots,
\]
the keyed-side homology maps through $\widetilde H_1(K_{12})$ and hence does not create a nontrivial boundary contribution across the cover.
\end{proof}

\begin{remark}[Scope]
Theorem~\ref{thm:HLJT-binary} is best viewed as a \emph{topological packaging} of the exact classical test. The value is not a new lossless criterion for $k=2$, but a bridge to multiway covers and to computational diagnostics (Section~\ref{sec:alg}).
\end{remark}

\subsection{Multiway Decompositions and the Nerve Obstruction}

For $k\ge 3$, losslessness is typically analyzed via chase or join-dependency characterizations \cite{aho1979,maier1983}. Topologically, a natural object is the \emph{nerve} of the cover of $K_F$ by induced subcomplexes.

\begin{definition}[Nerve of a decomposition cover]
Let $\{U_i\}_{i=1}^k$ be a decomposition cover of $U$ and let $K_i=K_F[U_i]$. The \emph{nerve} $\mathcal{N}$ is the simplicial complex on vertices $\{1,\dots,k\}$ where $\{i_0,\dots,i_p\}$ is a simplex iff
\[
K_{i_0}\cap\cdots\cap K_{i_p}\neq \emptyset.
\]
\end{definition}

\begin{proposition}[Cycle in the nerve is an obstruction to join-tree structure]\label{prop:nerve}
If $H_1(\mathcal{N})\neq 0$, then the cover cannot have a join-tree (tree-shaped nerve). In particular, this certifies the presence of cyclic cover-interaction (a structural hallmark of cyclic join behavior).
\end{proposition}

\begin{proof}
A join-tree cover has a nerve that is a tree (contractible 1-skeleton), hence $H_1(\mathcal{N})=0$. Contraposition yields the claim.
\end{proof}

\begin{remark}[Where this can be ``sharper'']
For $k\ge 3$, many practical heuristics check only pairwise intersections (or local key tests) and miss global cyclic interaction. The nerve homology detects \emph{global} cover cycles even when each pairwise intersection is nonempty. This does not replace the chase (which is exact), but it provides a fast structural warning that a join-tree style plan or acyclic-scheme properties cannot hold.
\end{remark}

\section{Algorithmic Consequences}\label{sec:alg}

\subsection{Computing Betti Numbers as a Schema Diagnostic}

Given $K_F$, compute boundary matrices $\partial_n$ and obtain Betti numbers by rank computations over a field:
\[
b_n = \dim \ker \partial_n - \dim \mathrm{im}\,\partial_{n+1}.
\]
In practice, one typically computes only $b_1$ (and sometimes $b_2$) since these capture the first obstructions to acyclicity. This yields:
\begin{itemize}
  \item \textbf{Cycle localization:} 1-cycles correspond to ``unfilled'' cyclic interactions among FDs in the canonical cover; these are candidates for redesign or decomposition.
  \item \textbf{Cover obstruction:} for a decomposition $\{U_i\}$, compute $H_1(\mathcal{N})$ of the nerve: $H_1(\mathcal{N})\neq 0$ is a fast certificate that the cover is not tree-like, hence one should not expect acyclic-scheme join simplifications.
\end{itemize}

\subsection{A Minimal Illustrative Diagnostic}

\begin{example}[Nerve cycle warning]
Let the decomposition have three components with pairwise overlaps, e.g.\ $U_1\cap U_2\neq\emptyset$, $U_2\cap U_3\neq\emptyset$, $U_3\cap U_1\neq\emptyset$, but empty triple overlap. Then the nerve $\mathcal{N}$ is (combinatorially) the 1-skeleton triangle, so
\[
H_1(\mathcal{N})\cong \mathbb{Z}
\]
is nontrivial. Figure~\ref{fig:nerve-cycle} depicts the situation: panel (a) shows the cover with all pairwise intersections but no triple overlap, and panel (b) shows the resulting cyclic nerve. This flags a global cycle in the cover-interaction even though all pairwise overlaps exist; in database terms, it is exactly the pattern that prevents a join-tree representation.
\end{example}
\FloatBarrier

\begin{figure}[!t]
\centering
\begin{minipage}[t]{0.48\textwidth}
\centering
\begin{tikzpicture}[scale=0.95]
  \draw[thick] (-0.5,0) circle (1.2);
  \draw[thick] (0.5,0) circle (1.2);
  \draw[thick] (0,0.8) circle (1.2);

  \node[fill=white,inner sep=1pt,yshift=-2pt] at (-1.55,-1.15) {$U_1$};
  \node[fill=white,inner sep=1pt,xshift=4pt,yshift=-2pt] at (1.45,-1.05) {$U_2$};
  \node[fill=white,inner sep=1pt,yshift=4pt] at (0,2.1) {$U_3$};

  \node[fill=white,inner sep=1pt] at (0,-0.15) {$U_1\cap U_2$};
  \node[fill=white,inner sep=1pt] at (-0.6,0.55) {$U_1\cap U_3$};
  \node[fill=white,inner sep=1pt] at (0.6,0.55) {$U_2\cap U_3$};

  \node[fill=white,inner sep=1pt] at (0,0.3) {$\emptyset$};
\end{tikzpicture}

\vspace{0.6em}
{\small (a) Decomposition cover with empty triple overlap.\strut}
\end{minipage}\hfill
\begin{minipage}[t]{0.48\textwidth}
\centering
\begin{tikzpicture}[scale=1]
  \node[circle,draw,thick,minimum size=16pt] (1) at (0,0) {$1$};
  \node[circle,draw,thick,minimum size=16pt] (2) at (2,0) {$2$};
  \node[circle,draw,thick,minimum size=16pt] (3) at (1,1.7) {$3$};

  \draw[thick] (1) -- (2) -- (3) -- (1);
\end{tikzpicture}

\vspace{0.6em}
{\small (b) Nerve $\mathcal{N}$ (triangle), so $H_1(\mathcal{N})\neq 0$.\strut}
\end{minipage}

\caption{A three-way decomposition whose nerve has a 1-cycle: all pairwise overlaps exist, but the triple overlap is empty, yielding a cyclic (non-tree) nerve.}
\label{fig:nerve-cycle}
\end{figure}
\FloatBarrier

What the figure adds beyond the verbal description is a clean algebraic certificate: the induced nerve has a nontrivial 1-cycle, so
\[
\widetilde H_1(\mathcal{N})\cong \mathbb{Z}\neq 0.
\]
This excludes the existence of any join tree for the cover (a tree-shaped nerve would force $\widetilde H_1(\mathcal{N})=0$), and it is precisely this failure of tree-likeness that underlies cyclic join behavior.

\section{A Topological View of Join Planning (Optional Perspective)}

Acyclic database schemes admit particularly efficient join evaluation and join ordering (e.g., via join trees) \cite{beeri1983,yannakakis1981}. Our contribution here is interpretive: when a decomposition cover has tree-like nerve (hence $H_1(\mathcal{N})=0$), one can evaluate joins along a corresponding join tree; when $H_1(\mathcal{N})\neq 0$, global cycles indicate multiple competing join paths and a larger optimization search space.

\subsection{A Worked Example: From Query Hypergraph to an Optimized Join Tree}

Consider the natural join query
\[
Q \,=\, R_1(A,B)\Join R_2(B,C)\Join R_3(C,D)\Join R_4(B,E)\Join R_5(E,F)\Join R_6(D,F).
\]
Its query hypergraph has vertices $\{A,B,C,D,E,F\}$ and hyperedges given by the attribute sets of the relations. A join tree is a tree whose nodes are the relations and where, for every attribute $X$, the set of relations containing $X$ induces a connected subtree (running-intersection property). Figure~\ref{fig:join-tree} shows one such join tree.

A standard evaluation strategy on an acyclic join tree is a two-pass semijoin reduction (Yannakakis-style) followed by a final join. Along an edge $R_i$--$R_j$ with separator $S=\mathrm{attrs}(R_i)\cap\mathrm{attrs}(R_j)$, define the semijoin
\[
R_i \ltimes_S R_j \;:=\; \{t\in R_i : \pi_S(t)\in \pi_S(R_j)\}.
\]
An optimized plan performs a bottom-up pass that replaces each relation by semijoin filters pushed from its neighbors (reducing intermediate sizes), then executes joins in a root-directed order.

\begin{example}[An explicit optimized join-tree plan]
Root the join tree at $R_2(B,C)$. In a bottom-up semijoin pass one may compute
\[
\begin{aligned}
R_1' &= R_1 \ltimes_{\{B\}} R_2, \\
R_4' &= R_4 \ltimes_{\{B\}} R_2, \\
R_3' &= R_3 \ltimes_{\{C\}} R_2, \\
R_5' &= R_5 \ltimes_{\{E\}} R_4', \\
R_6' &= R_6 \ltimes_{\{D\}} R_3' \;\ltimes_{\{F\}} R_5'.
\end{aligned}
\]
A subsequent join phase can then follow the tree structure, for instance
\[
Q \,=\, ((((R_1'\Join R_2)\Join R_3')\Join R_6')\Join R_5')\Join R_4'.
\]
Even with a simple cardinality model, pushing semijoins along separators (small intersections such as $\{B\}$, $\{C\}$, $\{E\}$, $\{D\}$, $\{F\}$) can substantially reduce the sizes of intermediate results compared to a naive left-deep plan.
\end{example}

\begin{figure}[!htbp]
\centering
\begin{tikzpicture}[
  node distance=14mm and 20mm,
  rel/.style={rectangle,draw,rounded corners,thick,inner sep=3pt,align=center},
  edge/.style={thick}
]
  \node[rel] (R2) {\begin{tabular}{c}$R_2$\\$(B,C)$\end{tabular}};
  \node[rel,above left=of R2] (R1) {\begin{tabular}{c}$R_1$\\$(A,B)$\end{tabular}};
  \node[rel,above right=of R2] (R3) {\begin{tabular}{c}$R_3$\\$(C,D)$\end{tabular}};
  \node[rel,below left=of R2] (R4) {\begin{tabular}{c}$R_4$\\$(B,E)$\end{tabular}};
  \node[rel,below right=of R4] (R5) {\begin{tabular}{c}$R_5$\\$(E,F)$\end{tabular}};
  \node[rel,below right=of R3] (R6) {\begin{tabular}{c}$R_6$\\$(D,F)$\end{tabular}};

  \draw[edge] (R1) -- node[midway,fill=white,inner sep=1pt] {$\{B\}$} (R2);
  \draw[edge] (R3) -- node[midway,fill=white,inner sep=1pt] {$\{C\}$} (R2);
  \draw[edge] (R4) -- node[midway,fill=white,inner sep=1pt] {$\{B\}$} (R2);
  \draw[edge] (R5) -- node[midway,fill=white,inner sep=1pt] {$\{E\}$} (R4);
  \draw[edge] (R6) -- node[midway,fill=white,inner sep=1pt] {$\{D\}$} (R3);
  \draw[edge] (R6) -- node[midway,fill=white,inner sep=1pt] {$\{F\}$} (R5);
\end{tikzpicture}
\caption{A join tree for the query $Q=\Join_{i=1}^6 R_i$ with edge labels denoting separators (shared attributes). Acyclicity (tree-shaped structure satisfying the running-intersection property) supports semijoin-based reduction and an optimized evaluation order along the tree.}
\label{fig:join-tree}
\end{figure}
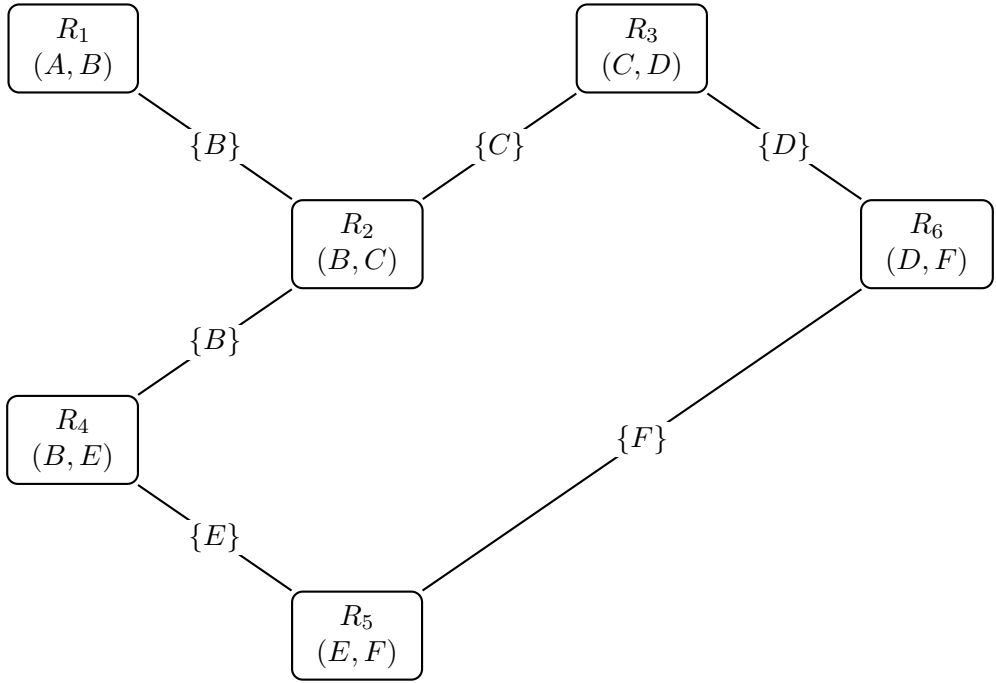
\FloatBarrier

In Fig.~\ref{fig:join-tree}, the edge labels are the separator sets $S_e$ that the semijoin operators in the example act on (e.g., $\{B\}$ between $R_1$ and $R_2$, $\{F\}$ between $R_5$ and $R_6$). This is the practical point of drawing the tree: it makes the interfaces explicit, so semijoin filters can be propagated along the edges and the final joins can be carried out in a root-directed order that respects the same structure.

\section{Conclusion}

We presented a simplicial-complex encoding of functional dependencies and used homological invariants as diagnostics for cyclic dependency structure. Simplicial Normal Form (SNF) was defined as vanishing reduced homology in positive degrees---a homological acyclicity condition intentionally separated from homotopy claims. For binary decompositions, we gave a topological reformulation of the classical lossless-join criterion and connected it to Mayer--Vietoris structure. For multiway decompositions, nerve homology provides a lightweight obstruction to join-tree structure, offering an algorithmic warning signal complementary to exact chase-based tests. These tools aim to bridge relational design intuition with computable topological invariants.



\begin{thebibliography}{99}

\bibitem{codd1970}
E.~F. Codd.
\newblock A relational model of data for large shared data banks.
\newblock \emph{Communications of the ACM}, 13(6):377--387, 1970.

\bibitem{armstrong1974}
W.~W. Armstrong.
\newblock Dependency structures of data base relationships.
\newblock In \emph{Information Processing 74}, pp. 580--583, 1974.

\bibitem{maier1983}
D.~Maier.
\newblock \emph{The Theory of Relational Databases}.
\newblock Computer Science Press, 1983.

\bibitem{ullman1988}
J.~D. Ullman.
\newblock \emph{Principles of Database and Knowledge-Base Systems}, Vol.~1.
\newblock Computer Science Press, 1988.

\bibitem{beeri1979}
C.~Beeri and P.~A. Bernstein.
\newblock Computational problems related to the design of normal form relational schemas.
\newblock \emph{ACM Transactions on Database Systems}, 4(1):30--59, 1979.

\bibitem{aho1979}
A.~V. Aho, C.~Beeri, and J.~D. Ullman.
\newblock The theory of joins in relational databases.
\newblock \emph{ACM Transactions on Database Systems}, 4(3):297--314, 1979.

\bibitem{beeri1983}
C.~Beeri, R.~Fagin, D.~Maier, and M.~Yannakakis.
\newblock On the desirability of acyclic database schemes.
\newblock \emph{Journal of the ACM}, 30(3):479--513, 1983.

\bibitem{yannakakis1981}
M.~Yannakakis.
\newblock Algorithms for acyclic database schemes.
\newblock In \emph{Proceedings of VLDB}, pp. 82--94, 1981.

\bibitem{hatcher2002}
A.~Hatcher.
\newblock \emph{Algebraic Topology}.
\newblock Cambridge University Press, 2002.

\end{thebibliography}
\end{document}